\documentclass[12pt,reqno]{amsart}

\usepackage{epsfig}
\usepackage{amsmath, amsthm, amsfonts, mathtools}
\usepackage{graphicx}

\numberwithin{equation}{section}

\newcommand{\C}{{\mathbb C}}

\newcommand{\Z}{{\mathbb Z}}

\newcommand{\al}{\alpha}
\newcommand{\be}{\beta}
\newcommand{\ga}{\gamma}
\newcommand{\Ga}{\Gamma}

\newcommand{\la}{\lambda}
\newcommand{\ep}{\varepsilon}

\newcommand{\De}{\Delta}

\newcommand{\sg}{\sigma}

\newcommand{\z}{\zeta}

\newtheorem{theo}{{\sc \bf Theorem}}[section]

\newtheorem{lem}[theo]{{\sc \bf Lemma}}
\newtheorem{prop}[theo]{{\sc \bf Proposition}}

\newenvironment{rem}{\medskip\noindent{\it Remark:\/} }{\medskip}

\newcommand{\Pf}{{\operatorname{Pf}\,}}

\usepackage{lipsum}
\makeatletter
\g@addto@macro{\endabstract}{\@setabstract}
\newcommand{\authorfootnotes}{\renewcommand\thefootnote{\@fnsymbol\c@footnote}}%
\makeatother

\begin{document}

\title[Dimer Model: Full Asymptotic Expansion]
{Dimer Model: Full Asymptotic Expansion of the Partition Function}

\begin{center}
  \maketitle \par \bigskip
  \normalsize
  \authorfootnotes
  Pavel Bleher\footnote{pbleher@iupui.edu}, Brad Elwood\footnote{bradelwood@gmail.com},
  Dra\v zen Petrovi\'c\footnote{petrovic.drazen@gmail.com}\par \bigskip

  Indiana University-Purdue University Indianapolis\par\bigskip
  
  \textit{\dedicatory{Dedicated to Ludvig Faddeev}}\par\bigskip

\end{center}


\date{\today}





\thanks{The first author is supported in part by the National Science Foundation (NSF) Grant DMS-1565602.}

\begin{abstract} We give a complete rigorous proof of the
full asymptotic expansion of the partition function of the dimer model on a square lattice
on a torus for general weights $z_h,z_v$ of the dimer model and arbitrary dimensions of the lattice $m,n$. We assume that $m$ is even and we show that the asymptotic expansion depends on the parity of $n$.  We review and extend the results of  Ivashkevich, Izmailian, and Hu \cite{IIH}
on the full asymptotic expansion of the partition function of the dimer model, 
and we give a rigorous estimate of the error term in the asymptotic expansion of the partition function.
\end{abstract}


\section{Introduction}

\subsection{Dimer Model on a Square Lattice}

We consider the dimer model on a square lattice $\Gamma_{m,n}=(V_{m,n},E_{m,n})$ on the torus $\Z_m\times \Z_n=\Z^2/(m\Z\times n\Z)$ (periodic boundary conditions), where
$V_{m,n}$ and $E_{m,n}$ are the sets of vertices and edges of $\Gamma_{m,n},$ respectively. 
A \textit{dimer} on $\Gamma_{m,n}$ is a set of two neighboring vertices $\langle x,y\rangle$ connected by an edge. 
A \textit{dimer configuration} $\sg$ on $\Ga_{m,n}$ is a set of dimers $\sg=\{\langle x_i,y_i\rangle,\;i=1,\ldots,\frac{mn}{2}\}$ which cover $V_{m,n}$ without overlapping. 
An example of a dimer configuration is shown in Fig. \ref{F8}. 
An obvious necessary condition for a configuration to exist is that at least one of $m,n$ is even, and so we assume that $m$ is even, $m=2m_0$.

\begin{figure}[h!]
\includegraphics[scale=.5]{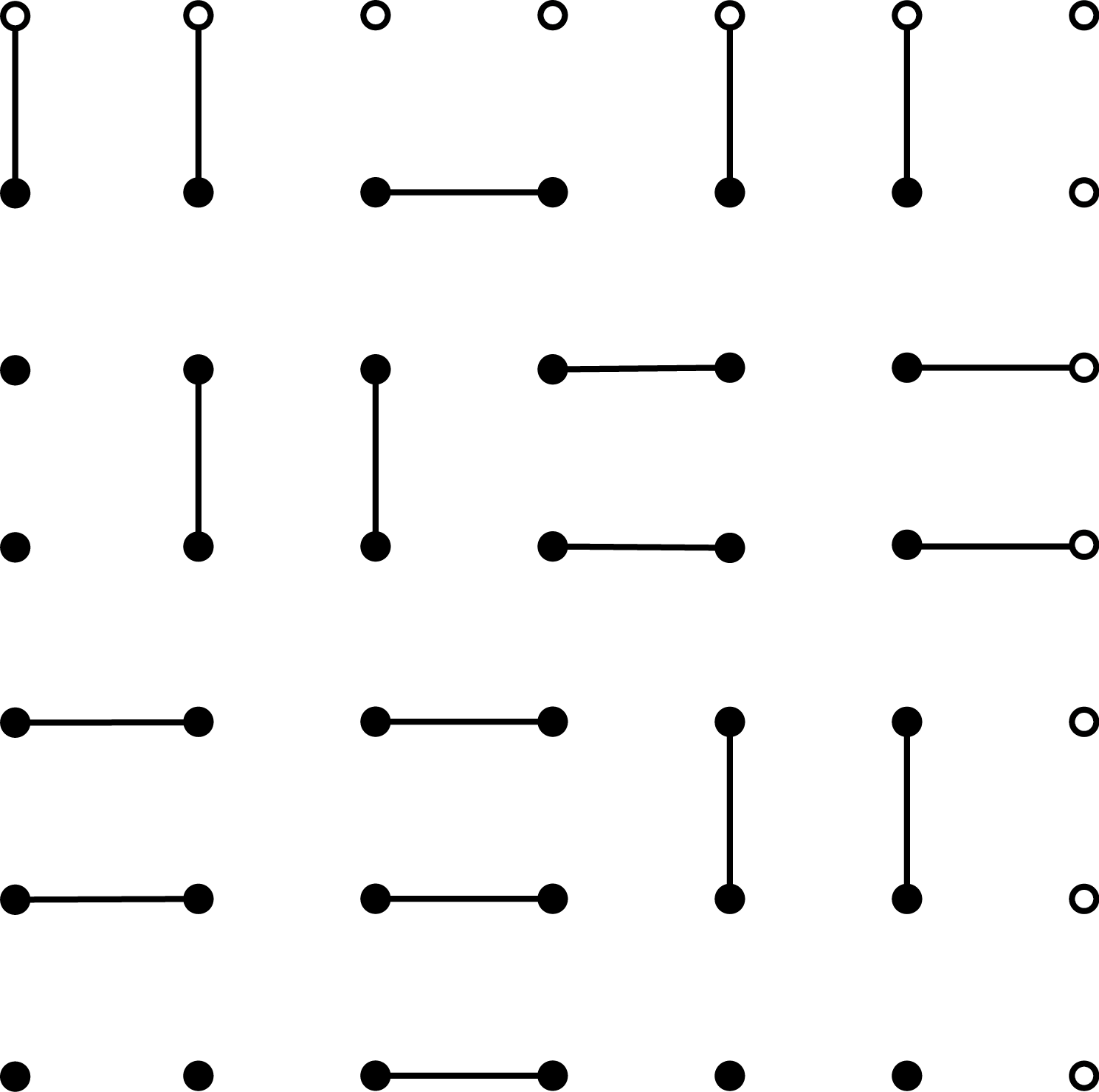}
\caption{Example of a dimer configuration on a square $6\times 6$ lattice on the torus.}
\label{F8}
\end{figure}

To define a weight of a dimer configuration, we split the full set of dimers in a configuration $\sg$ into two classes: horizontal and vertical, with respective weights $z_h, z_v>0.$ 
If we denote the total number of horizontal and vertical dimers in $\sg$ by $N_h(\sg)$ and $N_v(\sg),$ respectively, then the \textit{dimer configuration weight} is
\begin{equation}\label{int1}
w(\sg)=\prod_{i=1}^{\frac{mn}{2}}w(x_i,y_i)=z_h^{N_h(\sg)} z_v^{N_v(\sg)},
\end{equation}
where $w(x_i,y_i)$ denotes the weight of the dimer $\langle x_i,y_i\rangle\in\sg$.
We denote by $\Sigma_{m,n}$ the set of all dimer configurations on $\Ga_{m,n}$. 
The \textit{partition function} of the dimer model is given by
\begin{equation}\label{int2}
Z=\sum_{\sg\in\Sigma_{m,n}}w(\sg).
\end{equation}
Notice that if all the weights are set equal to one, then $Z$ simply counts the number of dimer configurations, or perfect matchings, on $\Gamma_{m,n}$. 

Our goal is to evaluate the full asymptotic series expansion of the partition function $Z$ as $m,n\to\infty$. The free energy of the dimer model on the square lattice was obtained in the papers of Kasteleyn \cite{Kas1} and Temperley and Fisher \cite{TempFish}. Our work is based on the
Kasteleyn's expression of the partition
function $Z$ on a torus as a linear combination of 4 Pfaffians developed in the works \cite{Kas1}, \cite{Kas2}, \cite{Kas3} (see also the works of Galluccio and Loebl \cite{GalluLoe}, Tesler \cite{Tes}, and  Cimasoni and Reshetikhin \cite{CimResh}). The constant term in the asymptotic of the partition function
was obtained by Ferdinand \cite{Ferdin} (see also the work
of Kenyon, Sun and Wilson \cite{KSW}). 

The asymptotic expansion of
the partition function on a torus was developed by Ivashkevich, Izmailian, and Hu \cite{IIH} and our calculations use their ideas. Ivashkevich, Izmailian, and Hu considered the case when $z_h=z_v$ and $n$ is even. In the present 
work we extend their calculations to arbitrary weights $z_h,z_v$
and to odd $n$. It is worth noticing that the asymptotic
expansions for even and odd values of $n$ are different. We give a 
complete rigorous proof of the asymptotic expansion of the partition
function, with an estimate of the error term.
The asymptotic expansion of the partition function is expressed in terms 
of the classical Jacobi theta functions, Dedekind eta function, and
Kronecker double series. The work \cite{IIH} has been further extended by 
Izmailian, Oganesyan, and Hu \cite{IOH} to the dimer model on a square lattice with various boundary conditions for both even and odd $n$. 
Our result for the dimer model on a torus coincides with the one in 
\cite{IOH} for even $n$, and for odd $n$ it  coincides  except for the value of the elliptic nome in formula \eqref{ntT6} below. The difference in the value of the elliptic nome for even and odd $n$ is explained after formula \eqref{eq4.4} in Section \ref{a3o} below.

It follows from \eqref{int1}, \eqref{int2} that the partition function $Z$ is a homogeneous polynomial of the variables $z_h,z_v$, and it can be written as
\begin{equation}\label{zeta1}
Z(z_h,z_v)=\sum_{\sigma\in \Sigma_{m,n}} z_h^{N_h(\sg)}z_v^{N_v(\sg)}
=z_h^{\frac{mn}{2}}Z(1,\zeta),
\end{equation} 
where
\begin{equation}\label{zeta2}
\zeta=\frac{z_v}{z_h}>0,
\end{equation}
so without loss of generality we may assume that
\begin{equation}\label{zeta3}
z_h=1,\quad z_v=\zeta,
\end{equation}
and we will evaluate the full asymptotic series expansion of the partition function 
$Z(1,\z)$ as $m,n\to\infty$.

To formulate our main result we have to introduce and remind some special functions and operators. 

\subsection{Function $g(x)$}\label{Fun_g}
Introduce the function
\begin{equation}\label{eq2.1}
g(x)=\ln\left(\z \sin(\pi x)+\sqrt{1+\z^2\sin^2(\pi x)}\right),
\end{equation}
where $\z>0$ is defined in \eqref{zeta2}.
Observe that $g(x)$ has the following properties:
\begin{enumerate} 
\item $g(-x)= -g(x),$ 
\item $g(x+1)=-g(x),$ 
\item $g(x)$ is real analytic on $[0,1]$ and 
\begin{equation}\label{eq2.2}
g\left(x\right)= \sum_{p=0}^\infty g_{2p+1}x^{2p+1}, 
\end{equation}
\noindent where
\begin{equation}\label{eq2.3}
g_1=\pi\z,\quad g_3=-\frac{\pi^3 \z (\z^2+1)}{6}\,,\quad g_5=\frac{\pi^5\z(\z^2+1)(9\z^2+1)}{120}\,,\;\ldots.
\end{equation}
\item $g(x)\ge C_0x$ on the segment $0\le x\le \frac{1}{2}\,$, with some $C_0>0$.
\end{enumerate}
The constant $C_0$ in the latter inequality can depend on $\z$. In what follows we assume that $\z$ is fixed and we do not indicate the dependence of various constants $C_k$  on $\z$. Unless otherwise is stated, the constants $C_k$ can be different in different inequalities.

Observe that since $g(x)$ is analytic at $x=0$, we have that
\begin{equation}\label{eq2.3a}
\left|g_{2p+1}\right|\le C \xi^p,
\end{equation}
with some $C,\xi>0$.

\subsection{Differential Operator $\De_p$}
Let $\mathcal S_p$ be the set of collections of positive integers \\ $(p_1,\ldots,p_r;q_1,\ldots,q_r)$, $1\le r\le p$, such that
\begin{equation}\label{eq2.4}
\begin{aligned}
\mathcal S_p=\left\{ 
(p_1,\ldots,p_r;q_1,\ldots,q_r)\;\big|\;
0<p_1<\ldots<p_r;\;  p_1q_1+\ldots+p_r q_r=p\right\}.
\end{aligned}
\end{equation}
Introduce the differential operator
\begin{equation}\label{eq2.5}
\begin{aligned}
\De_p=\sum_{\mathcal S_p}
\frac{(g_{2p_1+1})^{q_1}\ldots (g_{2p_r+1})^{q_r}}{q_1!\ldots q_r!}\,\frac{d^q}{d\la^q}\,,\quad
q=q_1+\ldots+q_r-1\,.
\end{aligned}
\end{equation}
Observe that 
\begin{equation}\label{eq2.6}
\Delta_1=g_3, \quad
\Delta_2=\frac{g^2_3}{2}\frac{d}{d\lambda}+g_5,\quad\Delta_3=\frac{g_3^3}{3!}\frac{d^2}{d\lambda^2}+g_3g_5\frac{d}{d\lambda}+g_7,\quad \ldots\ .
\end{equation}

\subsection{Kronecker's Double Series} The Kronecker double series of order $p$ with parameters $\al,\be$ is defined as
\begin{equation}\label{eq2.7a}
\begin{aligned}
K^{\al,\be}_p(\tau)=-\frac{p!}{(-2\pi i)^p}\sum_{(j,k)\not=(0,0)}\frac{e(j\al+k\be)}{(k+\tau j)^p}\,,
\end{aligned}
\end{equation}
where
\begin{equation}\label{eq2.7b}
\begin{aligned}
e(x)=e^{-2\pi ix}.
\end{aligned}
\end{equation}
We will use the following Kronecker double series with parameters $(\al,\be)=(\frac{1}{2},\frac{1}{2}), (0,\frac{1}{2}), (\frac{1}{2},0)$, respectively:
\begin{equation}\label{eq2.7}
\begin{aligned}
K^{\frac{1}{2},\frac{1}{2}}_p(\tau)=&-\frac{p!}{(-2\pi i)^p}\sum_{(j,k)\not=(0,0)}\frac{(-1)^{j+k}}{(k+\tau j)^p},\\
K^{0,\frac{1}{2}}_p(\tau)=&-\frac{p!}{(-2\pi i)^p}\sum_{(j,k)\not=(0,0)}\frac{(-1)^{k}}{(k+\tau j)^p},\\
K^{\frac{1}{2},0}_p(\tau)=&-\frac{p!}{(-2\pi i)^p}\sum_{(j,k)\not=(0,0)}\frac{(-1)^{j}}{(k+\tau j)^p}\,.
\end{aligned}
\end{equation}
We will use it for $\tau$ pure imaginary and $p\ge 4$. Then the double series are absolutely convergent.

\subsection{Dedekind Eta Function} The Dedekind eta function is defined as
\begin{equation}\label{eq2.8}
\begin{aligned}
\eta=\eta(\tau)=e^{\frac{\pi i\tau}{12}}\prod_{k=1}^\infty \left(1-e^{2\pi i\tau k}\right)
=q^{\frac{1}{12}}\prod_{k=1}^\infty \left(1-q^{2k}\right) ,
\end{aligned}
\end{equation}
where 
\begin{equation}\label{eq2.9}
q=e^{\pi i\tau}
\end{equation}
is the elliptic nome.

\subsection{Jacobi Theta Functions}
There are four Jacobi theta functions: 
\begin{equation}\label{eq2.10}
\begin{aligned}
\theta_1(z,q)&=2\sum_{k=0}^{\infty}(-1)^k q^{\left(k+\frac{1}{2}\right)^2} \sin\big((2k+1)z\big), \\
\theta_2(z,q)&=2\sum_{k=0}^{\infty}q^{\left(k+\frac{1}{2}\right)^2}\cos\big((2k+1)z\big), \\
\theta_3(z,q)&=1+2\sum_{k=1}^{\infty}q^{k^2}\cos(2k z), \\
\theta_4(z,q)&=1+2\sum_{k=1}^{\infty}(-1)^k q^{k^2} \cos(2k z),
\end{aligned}
\end{equation}
where $q=e^{\pi i\tau}$ is elliptic nome.

We have the following identities (see, e.g., \cite{Weber}):
\begin{equation}\label{dj}
\begin{aligned}
&\theta_2=\theta_2(0,q)=\frac{2\eta^2(2\tau)}{\eta(\tau)}\,,\\
&\theta_3=\theta_3(0,q)=\frac{\eta^5(\tau)}{\eta^2(2\tau)\eta^2(\frac{\tau}{2})}\,,\\
&\theta_4=\theta_4(0,q)=\frac{\eta^2(\frac{\tau}{2})}{\eta(\tau)}\,.
\end{aligned}
\end{equation}

Also (see, e.g., \cite{IIH}),
\begin{equation}\label{KtoT12}
\begin{aligned}
K_4^{0,\frac{1}{2}}(\tau)&=\frac{1}{30}\left(\frac{7}{8}\theta_2^8-\theta_3^4\theta_4^4\right)\,,\\
K_4^{\frac{1}{2},0}(\tau)&=\frac{1}{30}\left(\frac{7}{8}\theta_4^8-\theta_2^4\theta_3^4\right)\,,\\
K_4^{\frac{1}{2},\frac{1}{2}}(\tau)&=\frac{1}{30}\left(\frac{7}{8}\theta_3^8+\theta_2^4\theta_4^4\right)\,.\\
\end{aligned}
\end{equation}

\section{Main Result: Full Asymptotic Expansion of the Dimer Model Partition Function}

\subsection{Pfaffians}
We would like to evaluate the asymptotic expansion of the dimer model partition function $Z$ on the square lattice, $\Ga_{m,n}$, of dimensions $m\times n$, with periodic boundary conditions where $m,n\to\infty$ under the assumption that there exist positive constants $C_2>C_1$ such that 
\begin{equation}\label{fbc1.9}
C_1\le\frac{m}{n}\le C_2.
\end{equation} 
As shown by Kasteleyn \cite{Kas1,Kas2,Kas3}, the partition function $Z$ can be written in terms of four Pfaffians as
\begin{equation}\label{tbc1}
Z=\frac{1}{2}\left(-\Pf A_1+\Pf A_2+\Pf A_3+\Pf A_4\right),
\end{equation}
where $A_1, A_2, A_3, A_4$ are the antisymmetric Kasteleyn matrices with periodic-periodic, periodic-antiperiodic, antiperiodic-periodic, and antiperiodic-antiperiodic boundary conditions, respectively. Their determinants are given by the double product formulae as
\begin{equation}\label{tbc2}
\begin{aligned}
&\det A_i=\prod_{j=0}^{\frac{m}{2}-1}\prod_{k=0}^{n-1} \left[4\left(\sin^2\frac{2\pi(j+\al_i)}{m}+ \z^2\sin^2\frac{2\pi(k+\be_i)}{n}
\right)\right],
\end{aligned}
\end{equation}
with
\begin{equation}\label{tbc3}
\begin{aligned}
(\al_1,\be_1)=(0,0),\quad (\al_2,\be_2)=(0,1/2),\quad (\al_3,\be_3)=(1/2,0),
\quad (\al_4,\be_4)=(1/2,1/2).
\end{aligned}
\end{equation}
These double product formulae are obtained by diagonalizing the matrices $A_i$ (see \cite{Kas1,McCoy,McCoyWu}). The Pfaffian of a square antisymmetric matrix $A$ is related to its determinant through the classical identity:
\begin{equation}\label{pfdet}
(\Pf A)^2=\det A.
\end{equation}
Observe that $\det A_1=0$ due to the factor $j=k=0$ in \eqref{tbc2}, hence
\begin{equation}\label{pfa1}
\Pf A_1=0,
\end{equation}
and for odd $n$, $\det A_2=0$, due to the factor $j=0$, $k=\frac{n-1}{2}\,$, hence
\begin{equation}\label{pfa2}
\Pf A_2=0, \quad \textrm{if $n$ is odd.}
\end{equation}
In addition, 
\begin{equation}\label{pfa3}
\Pf A_3=\Pf A_4, \quad \textrm{if $n$ is odd}
\end{equation}
(see \cite{BEP}).
As shown in \cite{KSW},\cite{BEP}, 
\begin{equation}
\begin{aligned}
&\Pf A_2>0,\quad \textrm{if $n$ is even},\\
&\Pf A_3>0,\quad \Pf A_4>0\quad \textrm{for all $n$},
\end{aligned}
\end{equation}
hence from \eqref{tbc2} we obtain that
\begin{equation}\label{pfai}
\begin{aligned}
&\Pf A_i=\prod_{j=0}^{\frac{m}{2}-1}\prod_{k=0}^{n-1} \left[4\left(\sin^2\frac{2\pi(j+\al_i)}{m}+ \z^2\sin^2\frac{2\pi(k+\be_i)}{n}
\right)\right]^{1/2}.
\end{aligned} 
\end{equation}
Combining \eqref{tbc1} with \eqref{pfa1}, \eqref{pfa2}, \eqref{pfa3},
we obtain that 
\begin{equation}\label{tbc1a}
\begin{aligned}
Z&=\frac{1}{2}\left(\Pf A_2+\Pf A_3+\Pf A_4\right),\quad \textrm{if $n$ is even},\\
Z&=\Pf A_3,\quad \textrm{if $n$ is odd}.
\end{aligned}
\end{equation}

\subsection{Main Result}
Before stating the main theorem, let us introduce some additional notations. Denote
\begin{equation}\label{ntT4}
S=mn,\qquad \nu=\frac{m}{n}\,.
\end{equation}
We set
\begin{equation}\label{ntT5}
\tau =\begin{dcases} i\z \nu ,& \text{if $n$ is even,}\\
\frac{i \z \nu}{2},& \text{if $n$ is odd,}\end{dcases}
\end{equation}
so that the elliptic nome is equal to
\begin{equation}\label{ntT6}
q=e^{\pi i\tau}=\begin{dcases}e^{-\pi\zeta \nu},& \text{if $n$ is even,}\\
e^{\frac{-\pi \z \nu}{2}},&  \text{if $n$ is odd.}\end{dcases}
\end{equation}
For brevity we also denote
\begin{equation}\label{ntT7}
\eta=\eta(\tau),\qquad \theta_{k}=\theta_{k}(0,q), \quad k=2,3,4,
\end{equation}
where $\eta(\tau)$ is the Dedekind eta function, and $\theta_k(z,q)$ are the Jacobi theta functions.
The main result is the following asymptotic expansion of the partition function $Z$ in powers of $S^{-1}$, derived by Ivashkevich et al. in \cite{IIH} in the case $\z=1$ and $n$ is even. We give a complete rigorous proof of the asymptotic expansion for any $\z>0$ and for $n$ both even and odd.

\begin{theo} \label{main_thmT_TBC}
If $n$ is even, then as $ m,n\to\infty$ under condition \eqref{fbc1.9}, we have that
\begin{equation}\label{tcT,1.2}
\begin{aligned}
&Z= e^{SF}\left(C^{(2)}e^{R^{(2)}}+C^{(3)}e^{R^{(3)}}+C^{(4)}e^{R^{(4)}}\right),
\end{aligned}
\end{equation}
where
\begin{equation}\label{tcT,1.3a}
\begin{aligned}
&F=\frac{1}{\pi}\int\limits\limits_0^\z \frac{\arctan x}{x}\,dx,
\end{aligned}
\end{equation}
\begin{equation}\label{tcT,1.3b}
\begin{aligned}
&C^{(2)}= \frac{\theta_4^2}{2\eta^2},\quad 
C^{(3)}=\frac{\theta_2^2}{2\eta^2},\quad 
C^{(4)}=\frac{\theta_3^2}{2\eta^2},
\end{aligned}
\end{equation}
and $R^{(j)}$, $j=2,3,4$, admit the asymptotic expansions
\begin{equation}\label{tcT,1.3c}
\begin{aligned}
R^{(j)}\sim\sum_{p=1}^{\infty}\frac{R_{p}^{(j)}}{S^p}\,,\quad j=2,3,4,
\end{aligned}
\end{equation}
with
\begin{equation}\label{tcT,1.3d}
\begin{aligned}
&R_{p}^{(j)}=-\frac{2^{2p+1}\nu^{p+1}}{p+1}\,\De_p \left[K_{2p+2}^{\be_j,\al_j}
\left(\frac{i \nu\la}{\pi}\right)\right]\bigg|_{\la=\pi\z}\,,
\end{aligned}
\end{equation}
where $\al_j,\be_j$ are defined in \eqref{tbc3}.
In particular, by \eqref{eq2.6} and \eqref{KtoT12},
\begin{equation}\label{tcT,1.3e}
\begin{aligned}
&R_{1}^{(2)}=-\frac{2\nu^2g_3}{15}\bigg(\frac{7}{8}\theta_4^8-\theta_2^4\theta_3^4\bigg)\,,
\quad
R_{1}^{(3)}=-\frac{2\nu^2g_3}{15}\bigg(\frac{7}{8}\theta_2^8-\theta_3^4\theta_4^4\bigg)\,,\\
&R_{1}^{(4)}=-\frac{2\nu^2g_3}{15}\bigg(\frac{7}{8}\theta_3^8+\theta_2^4\theta_4^4\bigg)\,.
\end{aligned}
\end{equation}
Furthermore, if $n$ is odd, then as $ m,n\to\infty$ under condition \eqref{fbc1.9}, we have that 
\begin{equation}\label{tcT,1.2a}
\begin{aligned}
&Z= Ce^{SF+R},
\end{aligned}
\end{equation}
where $F$ is given in \eqref{tcT,1.3a},
\begin{equation}\label{tcT,1.3f}
\begin{aligned}
C=\frac{\theta_2}{\eta},
\end{aligned}
\end{equation}
and $R$  admits the asymptotic expansions
\begin{equation}\label{tcT,1.3g}
\begin{aligned}
R\sim\sum_{p=1}^{\infty}\frac{R_{p}}{S^p}\,,
\end{aligned}
\end{equation}
with
\begin{equation}\label{tcT,1.3h}
\begin{aligned}
R_{p}=-\frac{\nu^{p+1}}{p+1}\,\De_p \left[K_{2p+2}^{0,\frac{1}{2}}\left(\frac{i \nu\la}{2\pi}\right)\right]\bigg|_{\la=\pi\z}\,.
\end{aligned}
\end{equation}
By \eqref{eq2.6} and \eqref{KtoT12},
\begin{equation}\label{tcT,1.3i}
\begin{aligned}
&R_1=-\frac{\nu^2g_3}{60}\bigg(\frac{7}{8}\theta_2^8-\theta_3^4\theta_3^4\bigg)\,.
\end{aligned}
\end{equation}
\end{theo}

\begin{rem} As noticed by Kasteleyn \cite{Kas1}, the free energy $F$ in 
\eqref{tcT,1.3a} can be expressed in terms of the Euler dilogarithm function
\begin{equation}\label{dilog1}
{\rm L}_2(z)=-\int_0^z\frac{\ln(1-s)\,ds}{s}
\end{equation}
as
\begin{equation}\label{dilog2}
F(\z)=(2i)^{-1}\big[{\rm L}_2(i\z)-{\rm L}_2(-i\z)\big]\,.
\end{equation}
\end{rem}

The proof of Theorem \ref{main_thmT_TBC} will be given in Sections \ref{a2e}--\ref{a3o}.

\section{Asymptotic behavior of $\Pf A_2$  for even $n$}\label{a2e}

Since $\sin^2(x+\pi)=\sin^2x$, we can rewrite $\Pf A_2$ in \eqref{pfai} for even $n$ as
\begin{equation}\label{eq8.2}
\begin{aligned}
\Pf A_2=\prod_{j=0}^{\frac{m}{2}-1}\prod_{k=0}^{\frac{n}{2}-1} 
\left[4\left(\sin^2\frac{2j\pi}{m}+ \z^2\sin^2\frac{(2k+1)\pi}{n}\right)\right].
\end{aligned}
\end{equation}
Using the Chebyshev type identity (see e.g. \cite{Kas1}),
\begin{equation}\label{eq9.2} 
\prod_{j=0}^{\frac{m}{2}-1} \left[4\left( u^2+ \sin^2\frac{2j\pi}{m}\right)\right]=
\left[\left(u+\sqrt{1+u^2}\right)^{\frac{m}{2}}-\left(-u+\sqrt{1+u^2}\right)^{\frac{m}{2}}\right]^2\,,
\end{equation}
equation \eqref{eq8.2} is reduced to
\begin{equation}\label{eq9.3}
\begin{aligned}
&\Pf A_2=\prod_{k=0}^{\frac{n}{2}-1} 
\left[\left(u_{k}+\sqrt{1+u_{k}^2}\right)^{\frac{m}{2}}-\left(-u_{k}+\sqrt{1+u_{k}^2}\right)^{\frac{m}{2}}\right]^2\,,
\end{aligned}
\end{equation}
where
\begin{equation}\label{eq9.4}
u_{k}=\z\sin(\pi x_{k})\ge0, \quad x_{k}=\frac{2k+1}{n}\,.
\end{equation}
Observe that
\begin{equation}\label{eq9.5}
 \left(u_{k}+\sqrt{1+u_{k}^2}\right)\left(-u_{k}+\sqrt{1+u_{k}^2}\right)=1\,,
\end{equation}
hence
\begin{equation}\label{eq9.6}
\Pf A_2=B^{(2)}_{m,n}C^{(2)}_{m,n}\,,
\end{equation}
where
\begin{equation}\label{eq9.7}
\begin{aligned}
&B^{(2)}_{m,n}=\prod_{k=0}^{\frac{n}{2}-1} \left(u_{k}+\sqrt{1+u_{k}^2}\right)^m,\\
&C^{(2)}_{m,n}=\prod_{k=0}^{\frac{n}{2}-1} \left[1-\frac{1}{\left(u_{k}+\sqrt{1+u_{k}^2}\right)^m}\right]^2\,.
\end{aligned}
\end{equation}
Respectively,
\begin{equation}\label{eq9.8}
\ln (\Pf A_2)=G^{(2)}_{m,n}+H^{(2)}_{m,n}\,,
\end{equation}
with
\begin{equation}\label{eq9.9}
\begin{aligned}
G^{(2)}_{m,n}&=m\sum_{k=0}^{\frac{n}{2}-1} \ln\left(u_{k}+\sqrt{1+u_{k}^2}\right)
=m\sum_{k=0}^{\frac{n}{2}-1}  g\left(\frac{2k+1}{n}\right),\\
H^{(2)}_{m,n}&=2\sum_{k=0}^{\frac{n}{2}-1} \ln\left[1-\frac{1}{\left(u_{k}+\sqrt{1+u_{k}^2}\right)^m}\right].\\
\end{aligned}
\end{equation}
The function $g(x)$, defined in \eqref{eq2.1}, is real analytic, and
we will evaluate an asymptotic series expansion of $G^{(2)}_{m,n}$ for large $n$ by using an Euler--Maclaurin type formula and the Bernoulli polynomials $B_k(x)$ (see \cite{AS} or Appendix \ref{appD}).

\subsection{Evaluation of $G_{m,n}^{(2)}$} 
\begin{lem}\label{lemG}
As $n,m\to\infty$ under condition \eqref{fbc1.9}, we have that $G^{(2)}_{m,n}$ admits the following asymptotic expansion:
\begin{equation}\label{Lem_G2.1}
G^{(2)}_{m,n}\sim SF+\frac{\ga}{6}
-m\sum_{p=1}^\infty\frac{B_{2p+2}\left(\frac{1}{2}\right)g_{2p+1}}{(p+1)\left(\frac{n}{2}\right)^{2p+1}}\,,\quad \ga=\pi\nu\z.
\end{equation}
\end{lem}
\begin{proof}
\noindent From \eqref{eq9.9} we have that
\begin{equation}\label{eq9.10}
G_{m,n}^{(2)}=m\,G_{n}^{(2)},\quad G_{n}^{(2)}=\sum_{k=0}^{\frac{n}{2}-1} g\left(\frac{2k+1}{n}\right)\,.
\end{equation}
Using the Euler-Maclaurin formula \eqref{em5}, we obtain that $G_n^{(2)}$ is expanded in the asymptotic series in powers of $\frac{1}{n}$ as
\begin{equation}\label{eq9.11}
\begin{aligned}
G_{n}^{(2)}&\sim \frac{n}{2}\int\limits\limits_0^1 g(x)\,dx
+\sum_{p=1}^\infty\frac{B_{p}\left(\frac{1}{2}\right)}{\left(\frac{n}{2}\right)^{p-1}p!}[g^{(p-1)}(1)-g^{(p-1)}(0)].
\end{aligned}
\end{equation}
From  \eqref{eq2.2} and the equation $g(x+1)=-g(x)$ we obtain that 
\begin{equation}\label{eq9.12}
\begin{aligned}
&g^{(2p)}(0)=g^{(2p)}(1)=0,\\
&g^{(2p+1)}(0)=(2p+1)! g_{2p+1}, \quad
g^{(2p+1)}(1)=-(2p+1)!g_{2p+1}.
\end{aligned}
\end{equation}
Now, \eqref{eq9.11} becomes 
\begin{equation}\label{eq9.13}
\begin{aligned}
G_{n}^{(2)}&\sim \frac{n}{2}\int\limits\limits_0^1 g(x)\,dx
+\sum_{p=1}^\infty\frac{B_{2p}\left(\frac{1}{2}\right)}{\left(\frac{n}{2}\right)^{2p-1}(2p)!}[g^{(2p-1)}(1)-g^{(2p-1)}(0)]\\
&\sim \frac{n}{2}\int\limits\limits_0^1 g(x)\,dx
-\sum_{p=1}^\infty\frac{2B_{2p}\left(\frac{1}{2}\right)(2p-1)!g_{2p-1}}{\left(\frac{n}{2}\right)^{2p-1}(2p)!}\\
&\sim \frac{n}{2}\int\limits\limits_0^1 g(x)\,dx
-\sum_{p=0}^\infty\frac{B_{2p+2}\left(\frac{1}{2}\right)g_{2p+1}}{(p+1)\left(\frac{n}{2}\right)^{2p+1}}.
\end{aligned}
\end{equation}
Substituting \eqref{eq9.13} into \eqref{eq9.10}, we obtain that
\begin{equation}\label{eq9.14}
\begin{aligned}
G_{m,n}^{(2)}&\sim \frac{mn}{2}\int\limits\limits_0^1 g(x)\,dx
-m\sum_{p=0}^\infty\frac{B_{2p+2}\left(\frac{1}{2}\right)g_{2p+1}}{(p+1)\left(\frac{n}{2}\right)^{2p+1}}\,.\\
\end{aligned}
\end{equation}
Since $B_2(\frac{1}{2})=-\frac{1}{12}$ and $g_1=\pi\z$, we obtain that
\begin{equation}\label{eq9.15}
\begin{aligned}
G_{m,n}^{(2)}\sim SF+\frac{\ga}{6}
-m\sum_{p=1}^\infty\frac{B_{2p+2}\left(\frac{1}{2}\right)g_{2p+1}}{(p+1)\left(\frac{n}{2}\right)^{2p+1}}\,,\quad \ga=\pi\nu\z.
\end{aligned}
\end{equation}
where 
\begin{equation}\label{eq9.16}
\begin{aligned}
F=\frac{1}{2}\int\limits_0^1\ln\left(\z\sin(\pi x)+\sqrt{1+\z^2\sin^2(\pi x)}\right)\,dx\,.
\end{aligned}
\end{equation}
As shown by Kasteleyn \cite{Kas1},
\begin{equation}\label{eq9.17}
\begin{aligned}
\frac{1}{2}\int\limits_0^1\ln\left(\z\sin(\pi x)+\sqrt{1+\z^2\sin^2(\pi x)}\right)\,dx=
\frac{1}{\pi}\int\limits\limits_0^\z \frac{\arctan x}{x}\,dx,
\end{aligned}
\end{equation}
hence Lemma \ref{lemG} follows.
\end{proof}

Next, we evaluate $H^{(2)}_{m,n}$ in \eqref{eq9.9}.

\subsection{Evaluation of $H_{m,n}^{(2)}$} 
\begin{lem}\label{lem3.2}
As $n,m\to\infty$ under condition \eqref{fbc1.9}, we have that $H^{(2)}_{m,n}$ admits the following asymptotic expansion:
\begin{equation}\label{Lem_H2.1}
H^{(2)}_{m,n}\sim A^{(2)}+B^{(2)},
\end{equation}
with
\begin{equation}\label{Lem_H2.2}
\begin{aligned}
&A^{(2)}=4\sum_{k=0}^{\infty}\ln\left(1-e^{- \ga \left(2k+1\right)}\right),\quad \ga=\pi\nu\z,\\
&B^{(2)}=m\sum_{p=1}^\infty\frac{B_{2p+2}\left(\frac{1}{2}\right)\,g_{2p+1}}{(p+1)\left(\frac{n}{2}\right)^{2p+1}}\,+\sum_{p=1}^\infty \frac{R^{(2)}_p}{S^p}\,,
\end{aligned}
\end{equation}
where
\begin{equation}\label{Lem_H2.3}
R_p^{(2)}=-\frac{2^{2p+1}\nu^{p+1}}{p+1}\,\De_p 
\left[K_{2p+2}^{\frac{1}{2},0}
\left(\frac{i \nu\la}{\pi}\right)\right]\bigg|_{\la=\pi\z}\,.
\end{equation}
\end{lem}
\begin{proof}
From \eqref{eq9.9}, \eqref{eq9.4}, and \eqref{eq2.1} we have that 
\begin{equation}\label{eq9.18}
H_{m,n}^{(2)}=2\sum_{k=0}^{\frac{n}{2}-1} \ln\left[1-e^{-m g (x_{k})}\right],\quad x_k=\frac{2k+1}{n}\,.
\end{equation}
Since $g(x)\ge C_0x$ on the segment $[0,\frac{1}{2}]$, for some $C_0>0$, we have that
\begin{equation}\label{eq9.19}
e^{-mg(x_{k})}\le e^{-C_0 \nu (2k+1)}\,,\quad \nu=\frac{m}{n}\,,
\end{equation}
hence the sum in \eqref{eq9.18} is estimated from above by a geometric series, and  for any $L>0$ there is $R>0$ such that
\begin{equation}\label{eq9.21}
H_{m,n}^{(2)}=4\sum_{k=0}^{R\ln n}  \ln\left[1-e^{-mg (x_{k})}\right]+\mathcal O(n^{-L}),
\end{equation}
so that in our calculations we can restrict $k$ to $k\le R\ln n$.\\

Following \cite{IIH}, let us expand the logarithm in \eqref{eq9.21} into the Taylor series
\begin{equation}\label{eq9.22}
H_{m,n}^{(2)}=-4\sum_{k=0}^{R\ln n}  \sum_{j=1}^\infty \frac{e^{-mjg (x_{k})}}{j}+\mathcal O(n^{-L}).
\end{equation}
Observe that
\begin{equation}\label{eq9.22a}
mjg (x_{k})\ge C_0(2k+1)j,
\end{equation}
hence the series in $j$ converges exponentially and for any $L>0$ there is $R>0$ such that 
\begin{equation}\label{eq9.22b}
H_{m,n}^{(2)}=-4\sum_{k=0}^{R\ln n}  \sum_{j=1}^{R\ln n} \frac{e^{-mjg (x_{k})}}{j}+\mathcal O(n^{-L}).
\end{equation}
Expanding now $g(x)$ into power series \eqref{eq2.2}, we obtain that
\begin{equation}\label{eq9.23}
\begin{aligned}
e^{-mjg(x_{k})}&=e^{-mj \pi\z x_{k}} \exp\left[-mj\left(\sum_{p=1}^\infty g_{2p+1} x^{2p+1}_{k}\right)\right].
\end{aligned}
\end{equation}
Since $S=mn$ and $\nu=\frac{m}{n}$, we have that
$n^{2p}=\frac{S^p}{\nu^p}\,$. Hence,
\begin{equation}\label{eq9.24}
mx_{k}^{2p+1}=m\left(\frac{2k+1}{n}\right)^{2p+1}=\frac{\left(2k+1\right)^{2p+1}\nu^{p+1}}{S^p}\,
\end{equation}
and
\begin{equation}\label{eq9.25}
\begin{aligned}
e^{-mjg(x_{k})}
&=e^{-(2k+1)j \ga } \exp\left[-j\sum_{p=1}^\infty \frac{ g_{2p+1}\left(2k+1\right)^{2p+1}\nu^{p+1}}{S^p}\right], \quad \ga=\pi\nu\z.
\end{aligned}
\end{equation}
Denote
\begin{equation}\label{eq9.26}
a_p=- \left(2k+1\right)j\nu\, g_{2p+1},\quad x=\frac{\left(2k+1\right)^2\nu}{S}\,.
\end{equation}
Then formula \eqref{eq9.25} simplifies to
\begin{equation}\label{eq9.27}
\begin{aligned}
e^{-mjg(x_{k})}
&=e^{-(2k+1)j \ga } \exp\left(\sum_{p=1}^\infty a_p x^p\right).
\end{aligned}
\end{equation}
Substituting this expression into \eqref{eq9.22b}, we obtain that
\begin{equation}\label{eq9.22b0}
H_{m,n}^{(2)}=-4\sum_{k=0}^{R\ln n}  \sum_{j=1}^{R\ln n} 
\frac{e^{-(2k+1)j\ga} }{j}\,
 \exp\left(\sum_{p=1}^\infty a_p x^p\right)
+\mathcal O(n^{-L}).
\end{equation}
Expanding the exponent into the Taylor series, we obtain that
\begin{equation}\label{eq9.28}
\begin{aligned}
\exp\left(\sum_{p=1}^\infty a_p x^p\right)
=1+\sum_{p=1}^\infty b_p x^p,
\end{aligned}
\end{equation}
with
\begin{equation}\label{eq9.29}
\begin{aligned}
b_p= \sum_{\mathcal S_p}
\frac{(a_{p_1})^{q_1}\ldots (a_{p_r})^{q_r}}{q_1!\ldots q_r!}\,,
\end{aligned}
\end{equation}
where $\mathcal S_p$  is defined in \eqref{eq2.4}
(see \cite{IIH} and Appendix \ref{appA}). Thus,
\begin{equation}\label{eq9.22b1}
H_{m,n}^{(2)}=-4\sum_{k=0}^{R\ln n}  \sum_{j=1}^{R\ln n} 
\frac{e^{-(2k+1)j\ga} }{j}\,
\left(1+\sum_{p=1}^\infty b_p x^p\right)
 +\mathcal O(n^{-L}),
\end{equation}
or
\begin{equation}\label{eq9.33}
H_{m,n}^{(2)}= A^{(2)}_n+B^{(2)}_n+\mathcal O(n^{-L}),
\end{equation}
where
\begin{equation}\label{eq9.33a}
\begin{aligned}
&A^{(2)}_n=-4\sum_{k=0}^{R\ln n}  \sum_{j=1}^{R\ln n} 
\frac{e^{-(2k+1)j \ga }}{j}\,,\\
&B^{(2)}_n=-4\sum_{k=0}^{R\ln n}  \sum_{j=1}^{R\ln n} \sum_{p=1}^\infty
\frac{e^{-(2k+1)j\ga} }{j}\,
 b_p x^p,\quad x=\frac{\left(2k+1\right)^2\nu}{S}\,.
\end{aligned}
\end{equation}
We have that
\begin{equation}\label{eq9.33b}
A^{(2)}_n=-4\sum_{k=0}^\infty  \sum_{j=1}^\infty 
\frac{e^{-(2k+1)j \ga }}{j}+\mathcal O(n^{-L})=
4\sum_{k=0}^{\infty} \ln(1-e^{-(2k+1)\ga})+\mathcal O(n^{-L}).
\end{equation}
We write now $B_n^{(2)}$ as
\begin{equation}\label{eq9.33c}
\begin{aligned}
B^{(2)}_n&=B^{(2)}_{n,K}+R^{(2)}_{n,K},\quad 
B^{(2)}_{n,K}=
-4\sum_{k=0}^{R\ln n}  \sum_{j=1}^{R\ln n} \sum_{p=1}^{K-1}
\frac{e^{-(2k+1)j\ga} }{j}\,b_p x^p,\\
 R^{(2)}_{n,K}&=-4\sum_{k=0}^{R\ln n}  \sum_{j=1}^{R\ln n} \sum_{p=K}^{\infty}
\frac{e^{-(2k+1)j\ga} }{j}\,b_p x^p,
\end{aligned}
\end{equation}
and we would like to estimate the error term $R^{(2)}_{n,K}$.
To that end we will prove the following lemma:

\begin{lem} \label{error_estimate} {\rm (Error term estimate)} Fix any $\ep>0$.
Then as $S\to \infty$,
\begin{equation}\label{sl1}
\begin{aligned}
R^{(2)}_{n,K}=\mathcal O(S^{-K(1-\ep)})\,.
\end{aligned}
\end{equation}
\end{lem}

\begin{rem} Remind that $S=mn=\nu n^2$, where $C_1\le \nu\le C_2$, hence $S\to\infty$ implies that $n\to\infty$.
\end{rem}

\begin{proof} 
Let us estimate $b_p$.
From \eqref{eq9.26} and \eqref{eq2.3a} we have that
\begin{equation}\label{sl6}
|a_p|=(2k+1)j\nu|g_{2p+1}|\le C_1 (2k+1)j \xi^p,
\end{equation}
hence
\begin{equation}\label{sl7}
\left| \sum_{p=1}^\infty  a_p z^p\right|
\le C_2 (2k+1)j|z|,\quad |z|\le (2\xi)^{-1}, \quad z\in\C.
\end{equation}
This implies that
\begin{equation}\label{sl8}
\left| \sum_{p=1}^\infty  a_p z^p\right|
\le C_2,\quad \textrm{if}\quad |z|\le \min\big\{(2\xi)^{-1},[(2k+1)j]^{-1}\big\}, \quad z\in\C,
\end{equation}
hence 
\begin{equation}\label{sl9}
\left|\exp\left( \sum_{p=1}^\infty  a_p z^p\right)\right|
\le C_3=e^{C_2} ,\quad  \textrm{if}\quad |z|\le \min\big\{(2\xi)^{-1},[(2k+1)j]^{-1}\big\}, \quad z\in\C.
\end{equation}
By the Cauchy integral formula, 
\begin{equation}\label{sl10}
\frac{f^{(p)}(0)}{p!}=\frac{1}{2\pi i}\oint_{|z|=\rho}\frac{f(z)\,dz}{z^{p+1}}\,,
\end{equation}
applied to $f(z)=\exp\left( \sum_{p=1}^\infty  a_p z^p\right)$ and
$\rho=\min\big\{(2\xi)^{-1},[(2k+1)j]^{-1}\big\}$,
it follows that
\begin{equation}\label{sl11}
|b_p|\le C_3 \big[(2k+1)j\big]^p \quad \textrm{if}\quad (2k+1)j\ge \xi,
\end{equation}
and
\begin{equation}\label{sl12}
|b_p|\le C_3 \xi^p \quad \textrm{if}\quad (2k+1)j\le \xi.
\end{equation}
Using these estimates of $b_p$, we will now prove \eqref{sl1}.

As $n\to\infty$, we may assume that $R\ln n>\xi$, and we partition $R^{(2)}_{n,K}$ as follows: 
\begin{equation}\label{sl14}
\begin{aligned}
R^{(2)}_{n,K}&=
\sum_{j,k:\,(2k+1)j\le \xi}  
\left(\sum_{p=K}^\infty
\frac{e^{-(2k+1)j\ga} }{j}\,
 |b_p| x^p\right)\\
 &+\sum_{j,k:\,j,k\le R\ln n;\;(2k+1)j> \xi } 
 \left(\sum_{p=K}^\infty \frac{e^{-(2k+1)j\ga} }{j}\,|b_p| x^p \right),
 \quad x=\frac{\left(2k+1\right)^2\nu}{S}\,.
 \end{aligned}
\end{equation}
In the first term there are only finitely many possible values of
$j$ and $k$, and by \eqref{sl12},
\begin{equation}\label{sl15}
\begin{aligned}
\sum_{j,k:\,(2k+1)j\le \xi}  
\left(\sum_{p=K}^\infty
\frac{e^{-(2k+1)j\ga} }{j}\,
 |b_p| x^p\right)\le \sum_{p=K}^\infty (C_4S^{-1})^p
 =\mathcal O(S^{-K(1-\ep)})\,.
 \end{aligned}
\end{equation}
Consider now the second term in \eqref{sl14}. Using estimate \eqref{sl11}, we obtain that
\begin{equation}\label{sl18}
\begin{aligned}
\sum_{j,k:\,j,k\le R\ln n;\;(2k+1)j> \xi }
\left(\sum_{p=K}^\infty
\frac{e^{-(2k+1)j\ga} }{j}\,
 |b_p| x^p \right)\le \sum_{p=K}^\infty \frac{c_p}{S^p}\,,
 \end{aligned}
\end{equation}
with
\begin{equation}\label{sl19}
\begin{aligned}
0<c_p\le 
C_6\sum_{k=0}^{R\ln n}  
\sum_{j=1}^{R\ln n}
e^{-(2k+1)j\ga}\big[(2k+1)^3j\nu\big]^p\le \big[C_7 \nu(R\ln n)^4\big]^p,
 \end{aligned}
\end{equation}
hence
\begin{equation}\label{sl20}
\begin{aligned}
\sum_{p=K}^\infty \frac{c_p}{S^p}\le \sum_{p=K}^\infty 
\left[\frac{C_7 (R\ln n)^4}{n^2}\right]^{p}=\mathcal O(S^{-K(1-\ep)})\,.
 \end{aligned}
\end{equation}
Thus,
\begin{equation}\label{sl21}
\begin{aligned}
\sum_{j,k:\,j,k\le R\ln n;\;(2k+1)j> \xi }
\left(\sum_{p=K}^\infty
\frac{e^{-(2k+1)j\ga} }{j}\,
 |b_p| x^p \right)=\mathcal O(S^{-K(1-\ep)})\,,
 \end{aligned}
\end{equation}
and \eqref{sl1} is proved.
\end{proof}

Next, we would like to replace $R\ln n$ in $B^{(2)}_{n,K}$ in \eqref{eq9.33c} by $\infty$. Denote
\begin{equation}\label{sl22}
\begin{aligned}
B^{(2)}(p)=
-4\sum_{k=0}^{\infty}  \sum_{j=1}^{\infty} 
\frac{e^{-(2k+1)j\ga} }{j}\,b_p x^p,\quad
B_n^{(2)}(p)=
-4\sum_{k=0}^{R\ln n}  \sum_{j=1}^{R\ln n} 
\frac{e^{-(2k+1)j\ga} }{j}\,b_p x^p.
\end{aligned}
\end{equation}
Then using estimate \eqref{sl11} of $b_p$, we obtain that
\begin{equation}\label{sl23}
\begin{aligned}
|B^{(2)}(p)-B^{(2)}_{n}(p)|&=
4\left|\sum_{j,k:\;j>0,\;k\ge 0;\;\max\{j,k\}>R\ln n}  \frac{e^{-(2k+1)j\ga} }{j}\,
b_p x^p\right|\\
&\le
C_0 S^{-p}\left|\sum_{j,k:\;j>0,\;k\ge 0;\;\max\{j,k\}>R\ln n}  \frac{e^{-(2k+1)j\ga} }{j}\,
[(2k+1)^3j\nu]^p \right|.
\end{aligned}
\end{equation}
We have that
\begin{equation}\label{sl24}
\begin{aligned}
\sum_{k=R\ln n}^\infty k^{3p} e^{-k\ga}\le C_1(p) n^{-R\ga/2},
\end{aligned}
\end{equation}
hence from \eqref{sl23} we obtain that
\begin{equation}\label{sl25}
\begin{aligned}
|B^{(2)}(p)-B^{(2)}_{n}(p)|\le C_2(p) n^{-R\ga/2}.
\end{aligned}
\end{equation}

From Lemma \ref {error_estimate}, \eqref{eq9.33b},  and \eqref{sl25} we obtain an 
asymptotic expansion of $H_{m,n}^{(2)}$ in powers of $S^{-1}$ as
\begin{equation}\label{eq9.34b}
H_{m,n}^{(2)}\sim A^{(2)}+B^{(2)},
\end{equation}
with
\begin{equation}\label{eq9.34c}
\begin{aligned}
&A^{(2)}=4\sum_{k=0}^{\infty} \ln(1-e^{-(2k+1)\ga})\,,
\quad B^{(2)}=\sum_{p=1}^\infty d_p S^{-p},
\end{aligned}
\end{equation}
where
\begin{equation}\label{eq9.34e}
\begin{aligned}
d_p=-4\nu^p\sum_{k=0}^{\infty}  \sum_{j=1}^{\infty}
\frac{e^{-(2k+1)j\ga} b_p(2k+1)^{2p}}{j}\,.
\end{aligned}
\end{equation}
Here $b_p$ is given by equation \eqref{eq9.29}, and it satisfies estimate
\eqref{sl12}, which shows that the series over $k,j$ in the latter formula is convergent. We can transform $d_p$ as follows.

Substituting expression \eqref{eq9.26} for $a_{p_1},\ldots,a_{p_r}\,$
into \eqref{eq9.29}, we obtain that
\begin{equation}\label{eq9.30}
\begin{aligned}
b_p=\sum_{\mathcal S_p}
\frac{(g_{2p_1+1})^{q_1}\ldots (g_{2p_r+1})^{q_r}\left[-\left(2k+1\right)j\nu \right]^{q_1+\ldots+q_r}}{q_1!\ldots q_r!}\,.
\end{aligned}
\end{equation}
We can simplify the latter expression using operator $\Delta_p$ in \eqref{eq2.5}.
Namely, we have that
\begin{equation}\label{eq9.31a}
\begin{aligned}
\De_p \left[e^{- \left(2k+1\right)j\nu\la}\right]\Big|_{\la=\pi\z}
&=e^{-(2k+1)j \ga }\\
&\times \sum_{\mathcal S_p}
\frac{(g_{2p_1+1})^{q_1}\ldots (g_{2p_r+1})^{q_r}\left[-(2k+1)j\nu \right]^{q_1+\ldots+q_r-1}}{q_1!\ldots q_r!}\,,
\end{aligned}
\end{equation}
hence
\begin{equation}\label{eq9.32}
\begin{aligned}
b_p=e^{ \left(2k+1\right)j\ga}\left[-(2k+1)j\nu \right]
\De_p \left[e^{- \left(2k+1\right)j\nu\la}\right]\Big|_{\la=\pi\z}\,.
\end{aligned}
\end{equation}
Returning back to formula \eqref{eq9.34e}, we obtain that
\begin{equation}\label{eq9.34d}
\begin{aligned}
d_p=4\nu^{p+1}\sum_{k=0}^{\infty}  \sum_{j=1}^{\infty}
 (2k+1)^{2p+1}\,\De_p \left[e^{- \left(2k+1\right)j\nu\la}\right]\Big|_{\la=\pi\z}\,.
\end{aligned}
\end{equation}
To relate $d_p$ to the  Kronecker double series, introduce the function
\begin{equation}\label{eq9.35}
\begin{aligned}
F^{(2)}_p(\la)=\sum_{k=0}^{\infty}  
\sum_{j=1}^\infty \left(k+\frac{1}{2}\right)^{2p+1} e^{-2\nu j \la \left(k+\frac{1}{2}\right)}\,.
\end{aligned}
\end{equation}
Then
\begin{equation}\label{eq9.36}
\begin{aligned}
d_p=2^{2p+3}\nu^{p+1} \De_p \big[F^{(2)}_p(\la)\big]\big|_{\la=\pi\z}\,.
\end{aligned}
\end{equation}
The function $F^{(2)}_p(\la)$ can be expressed in terms of the Kronecker double series of a complex argument.
More precisely, from equation \eqref{kr10} in Appendix \ref{appF} we have that  
\begin{equation}\label{eq9.37}
F^{(2)}_p(\la)=\frac{B_{2p+2}\left(\frac{1}{2}\right)-K_{2p+2}^{\frac{1}{2},0}\left(\frac{i\nu\la}{\pi}\right)}{4(p+1)}.
\end{equation}
Furthermore, since the free term in the operator $\De_p$ in \eqref{eq2.5} is equal to $g_{2p+1}$, we obtain that
\begin{equation}\label{eq9.38}
\De_p F^{(2)}_p(\la)=\frac{B_{2p+2}\left(\frac{1}{2}\right)}{4(p+1)}g_{2p+1}-\frac{\De_p K_{2p+2}^{\frac{1}{2},0}\left(\frac{i\nu\la}{\pi}\right)}{4(p+1)},
\end{equation}
and therefore,
\begin{equation}\label{eq9.39}
\begin{aligned}
B^{(2)}&=\sum_{p=1}^\infty \frac{2^{2p+1}\nu^{p+1}B_{2p+2}\left(\frac{1}{2}\right)}{S^p(p+1)}g_{2p+1}-
\sum_{p=1}^\infty \frac{2^{2p+1}\nu^{p+1}}{S^p(p+1)}
\De_p 
\left[K_{2p+2}^{\frac{1}{2},0}
\left(\frac{i \nu\la}{\pi}\right)\right]\bigg|_{\la=\pi\z}\\
&=m\sum_{p=1}^\infty\frac{B_{2p+2}\left(\frac{1}{2}\right)}{(p+1)\left(\frac{n}{2}\right)^{2p+1}}g_{2p+1}+\sum_{p=1}^\infty \frac{R^{(2)}_p}{S^p}\,,
\end{aligned}
\end{equation}
(recall that $S=\nu n^2$),
and this completes the proof of Lemma \ref{lem3.2}.
\end{proof}
\subsection{Evaluation of $ \ln (\Pf A_2)$}
\text{}\\

\noindent To evaluate $\ln (\Pf A_2)$, substitute \eqref{Lem_G2.1} and \eqref{Lem_H2.1} into \eqref{eq9.8} to obtain
\begin{equation}\label{eq9.41}
\begin{aligned}
\ln (\Pf A_2)&\sim\bigg(SF+\frac{\ga}{6}
-m\sum_{p=1}^\infty\frac{B_{2p+2}\left(\frac{1}{2}\right)}{(p+1)\left(\frac{n}{2}\right)^{2p+1}}g_{2p+1}\bigg)\\
&+\bigg(4\sum_{k=0}^{\infty} \ln(1-e^{-(2k+1)\ga})+m\sum_{p=1}^\infty\frac{B_{2p+2}\left(\frac{1}{2}\right)}{(p+1)\left(\frac{n}{2}\right)^{2p+1}}g_{2p+1}+\sum_{p=1}^\infty \frac{R^{(2)}_p}{S^p}\bigg)\\
&= SF+\frac{\ga}{6}+\sum_{k=1}^{\infty}\ln\left(1-e^{- \left(2k-1\right)\ga}\right)^{4}+\sum_{p=1}^\infty \frac{R^{(2)}_p}{S^p}\,,
\quad \ga=\pi\nu\z.
\end{aligned}
\end{equation}
Note that the series containing the Bernoulli polynomials $B_{2p+2}(\frac{1}{2})$ cancel out. 
If we let the elliptic nome be equal to 
\begin{equation}\label{eq9.42}
q=e^{-\ga}\,,
\end{equation}
then by using \eqref{eq2.8} and \eqref{dj}, 
we obtain that
\begin{equation}\label{eq9.43}
\begin{aligned}
e^\frac{\ga}{6}\prod_{k=1}^\infty(1-e^{-(2k-1)\ga})^4&=q^{-1/6}\prod_{k=1}^{\infty}(1-q^{2k-1})^4=\left[q^{-1/24}\prod_{k=1}^{\infty}\frac{(1-q^{k})}{(1-q^{2k})}\right]^4\\
&=\left[\frac{q^{1/24}\prod_{k=1}^{\infty}(1-q^{k})}{q^{1/12}\prod_{k=1}^{\infty}(1-q^{2k})}\right]^4
=\left[\frac{\eta\left(\frac{\tau}{2}\right)}{\eta(\tau)}\right]^4=\frac{\theta_4^2}{\eta^2}\,.
\end{aligned}
\end{equation}
Therefore, equation \eqref{eq9.41} implies that
\begin{equation}\label{eq9.44}
\Pf A_2=e^{SF}\frac{\theta_4^2}{\eta^2}\,e^{R^{(2)}},
\quad R^{(2)}\sim 
\sum_{p=1}^\infty \frac{R^{(2)}_p}{S^p}\,.
\end{equation}

\section{Asymptotic expansions of $\Pf A_3$ and $\Pf A_4$ for even $n$}\label{a3e}

The asymptotic expansions of  $\Pf A_3$ and $\Pf A_4$
for even $n$ can be obtained in the same way as the one of $\Pf A_2$. 
Let us briefly discuss them.

From formula \eqref{pfai} we have that
\begin{equation}\label{pfa3a}
\begin{aligned}
&\Pf A_3=\prod_{j=0}^{\frac{m}{2}-1}\prod_{k=0}^{n-1} \left[4\left(\sin^2\frac{2\pi(j+\frac{1}{2})}{m}+ \z^2\sin^2\frac{2\pi k}{n}
\right)\right]^{1/2}.
\end{aligned} 
\end{equation}
Using that $\sin^2(x+\pi)=\sin^2x$, we can rewrite the latter formula for even $n$ as
\begin{equation}\label{pfa3b}
\begin{aligned}
&\Pf A_3=\prod_{j=0}^{\frac{m}{2}-1}\prod_{k=0}^{\frac{n}{2}-1} \left[4\left(\sin^2\frac{2\pi(j+\frac{1}{2})}{m}+ \z^2\sin^2\frac{2\pi k}{n}
\right)\right].
\end{aligned} 
\end{equation}
Using the Chebyshev type identity (see \cite{Kas1}),
\begin{equation}\label{eq3.2} 
\prod_{j=0}^{\frac{m}{2}-1} \left[4\left( u^2+ \sin^2\frac{(2j+1)\pi}{m}\right)\right]=
\left[\left(u+\sqrt{1+u^2}\right)^{\frac{m}{2}}+\left(-u+\sqrt{1+u^2}\right)^{\frac{m}{2}}\right]^2\,,
\end{equation}
we obtain that 
\begin{equation}\label{eq3.3}
\begin{aligned}
&\Pf A_3=\prod_{k=0}^{\frac{n}{2}-1} 
\left[\left(u_{k}+\sqrt{1+u_{k}^2}\right)^{\frac{m}{2}}+\left(-u_{k}+\sqrt{1+u_{k}^2}\right)^{\frac{m}{2}}\right]^2\,.
\end{aligned}
\end{equation}
where
\begin{equation}\label{eq3.4}
u_{k}=\z\sin(\pi x_{k})\ge0, \quad x_{k}=\frac{2k}{n}\,,
\end{equation}
which implies that
\begin{equation}\label{eq3.8}
\ln (\Pf A_3)=G^{(3)}_{m,n}+H^{(3)}_{m,n}\,,
\end{equation}
where
\begin{equation}\label{eq3.9}
\begin{aligned}
&G^{(3)}_{m,n}=m\sum_{k=0}^{\frac{n}{2}-1} \ln\left(u_{k}+\sqrt{1+u_{k}^2}\right)
=m\sum_{k=0}^{\frac{n}{2}-1}  g\left(\frac{2k}{n}\right),\\
&H^{(3)}_{m,n}=2\sum_{k=0}^{\frac{n}{2}-1} \ln\left[1+\frac{1}{\left(u_{k}+\sqrt{1+u_{k}^2}\right)^m}\right].\\
\end{aligned}
\end{equation}
Using the Euler--Maclaurin formula, we obtain  the following asymptotic expansion:
\begin{equation}\label{Lem_G2.2}
G^{(3)}_{m,n}\sim SF-\frac{\ga}{3}
-m\sum_{p=1}^\infty\frac{B_{2p+2}\left(0\right)g_{2p+1}}{(p+1)\left(\frac{n}{2}\right)^{2p+1}}\,, \quad \ga=\pi\nu\z.
\end{equation}
Next, we obtain an asymptotic expansion of $H^{(3)}_{m,n}$:
\begin{equation}\label{Lem_H3.1}
H^{(3)}_{m,n}\sim A^{(3)}+B^{(3)},
\end{equation}
with
\begin{equation}\label{Lem_H3.2}
\begin{aligned}
&A^{(3)}=4\sum_{k=1}^{\infty}\ln\left(1+e^{-2 k\ga}\right)+2\ln2,
\\
&B^{(3)}=m\sum_{p=1}^\infty\frac{B_{2p+2}\left(0\right)}{(p+1)\left(\frac{n}{2}\right)^{2p+1}}g_{2p+1}+\sum_{p=1}^\infty \frac{R^{(3)}_p}{S^p}\,,
\end{aligned}
\end{equation}
where
\begin{equation}\label{Lem_H3.3}
R_p^{(3)}=-\frac{2^{2p+1}\nu^{p+1}}{p+1}\,\De_p 
\left[K_{2p+2}^{0,\frac{1}{2}}
\left(\frac{i \nu\la}{\pi}\right)\right]\bigg|_{\la=\pi\z}\,.
\end{equation}
Substituting \eqref{Lem_G2.2} and \eqref{Lem_H3.1} into \eqref{eq3.8} we obtain that
\begin{equation}\label{eq3.41}
\begin{aligned}
\ln (\Pf A_3)&\sim
 SF-\frac{\ga}{3}+4\sum_{k=1}^{\infty}\ln\left(1+e^{-2 k\ga}\right)+2\ln2+\sum_{p=1}^\infty \frac{R^{(3)}_p}{S^p}\,.
\end{aligned}
\end{equation}
Let $q=e^{-\ga}$. Since
\begin{equation}\label{eq3.43}
\begin{aligned}
4e^{-\frac{\ga}{3}}\prod_{k=1}^\infty \left(1+e^{-2 k\ga}\right)^4&=
4q^{1/3}\prod_{k=1}^{\infty}(1+q^{2k})^4=4\left[q^{1/12}\prod_{k=1}^{\infty}\frac{(1-q^{4k})}{(1-q^{2k})}\right]^4\\
&=4\left[\frac{q^{1/6}\prod_{k=1}^{\infty}(1-q^{4k})}{q^{1/12}\prod_{k=1}^{\infty}(1-q^{2k})}\right]^4=4\left[\frac{\eta(2 \tau)}{\eta(\tau)}\right]^4=\frac{\theta_2^2}{\eta^2}\,,
\end{aligned}
\end{equation}
we obtain that
\begin{equation}\label{eq3.44}
\Pf A_3=e^{SF}\frac{\theta_2^2}{\eta^2}\,e^{R^{(3)}},
\quad R^{(3)}\sim 
\sum_{p=1}^\infty \frac{R^{(3)}_p}{S^p}\,.
\end{equation}

Let us turn to  $\Pf A_4$.
Since $\sin^2(x+\pi)=\sin^2x$, we can rewrite $\Pf A_4$ in \eqref{pfai} for even $n$ as
\begin{equation}\label{eq5.2}
\begin{aligned}
\Pf A_4&=\prod_{j=0}^{\frac{m}{2}-1}\prod_{k=0}^{\frac{n}{2}-1} \left[4\left(\sin^2\frac{(2j+1)\pi}{m}+\z^2\sin^2\frac{(2k+1)\pi}{n}\right)\right]\,.
\end{aligned}
\end{equation}
\noindent Using  identity \eqref{eq3.2},
we obtain that
\begin{equation}\label{eq6.3}
\begin{aligned}
&\Pf A_4=\prod_{k=0}^{\frac{n}{2}-1} 
\left[\left(u_{k}+\sqrt{1+u_{k}^2}\right)^{\frac{m}{2}}+\left(-u_{k}+\sqrt{1+u_{k}^2}\right)^{\frac{m}{2}}\right]^2\,,
\end{aligned}
\end{equation}
where
\begin{equation}\label{eq6.4}
u_{k}=\z\sin(\pi x_{k})\ge0, \quad x_{k}=\frac{2k+1}{n}\,,
\end{equation}
which implies that
\begin{equation}\label{eq6.8}
\ln (\Pf A_4)=G^{(4)}_{m,n}+H^{(4)}_{m,n}\,,
\end{equation}
where
\begin{equation}\label{eq6.9}
\begin{aligned}
&G^{(4)}_{m,n}=m\sum_{k=0}^{\frac{n}{2}-1} \ln\left(u_{k}+\sqrt{1+u_{k}^2}\right)
=m\sum_{k=0}^{\frac{n}{2}-1}  g\left(\frac{2k+1}{n}\right),\\
&H^{(4)}_{m,n}=2\sum_{k=0}^{\frac{n}{2}-1} \ln\left[1+\frac{1}{\left(u_{k}+\sqrt{1+u_{k}^2}\right)^m}\right].\\
\end{aligned}
\end{equation}
Using the Euler--Maclaurin formula, we obtain the
following asymptotic expansion:
\begin{equation}\label{Lem_G4.2}
G^{(4)}_{m,n}\sim SF+\frac{\ga}{6}
-m\sum_{p=1}^\infty\frac{B_{2p+2}\left(\frac{1}{2}\right)g_{2p+1}}{(p+1)\left(\frac{n}{2}\right)^{2p+1}}\,, \quad \ga=\pi\nu\z,
\end{equation}
and then, similar to Lemma \ref{lem3.2}, we obtain   that
\begin{equation}\label{Lem_H4.1}
H^{(4)}_{m,n}\sim A^{(4)}+B^{(4)},
\end{equation}
with
\begin{equation}\label{Lem_H4.2}
\begin{aligned}
&A^{(4)}=4\sum_{k=0}^{\infty}\ln\left(1+e^{-\ga (2k+1)}\right),\\
&B^{(4)}=m\sum_{p=1}^\infty\frac{B_{2p+2}\left(\frac{1}{2}\right)}{(p+1)\left(\frac{n}{2}\right)^{2p+1}}g_{2p+1}+\sum_{p=1}^\infty \frac{R^{(4)}_p}{S^p}\,,
\end{aligned}
\end{equation}
where
\begin{equation}\label{Lem_H4.3}
R_p^{(4)}=-\frac{2^{2p+1}\nu^{p+1}}{p+1}\,\De_p 
\left[K_{2p+2}^{\frac{1}{2},\frac{1}{2}}
\left(\frac{i \nu\la}{\pi}\right)\right]\bigg|_{\la=\pi\z}\,.
\end{equation}
Substituting \eqref{Lem_G4.2} and \eqref{Lem_H4.1} into \eqref{eq6.8},
we obtain that
\begin{equation}\label{eq4}
\Pf A_4=e^{SF}\frac{\theta_3^2}{\eta^2}\,e^{R^{(4)}},
\quad R^{(4)}\sim 
\sum_{p=1}^\infty \frac{R^{(4)}_p}{S^p}\,.
\end{equation}
Substituting equations \eqref{eq9.44}, \eqref{eq3.44}, and \eqref{eq4} into \eqref{tbc1a},  we obtain
the asymptotic formula for $Z$, \eqref{tcT,1.2}, for even $n$.



\section{Asymptotic behavior of $\Pf A_3$ for odd $n$}\label{a3o}

From equation \eqref{pfai} we have that
\begin{equation}\label{tbc2.2}
\begin{aligned}
\Pf A_3=\prod_{j=0}^{\frac{m}{2}-1}\prod_{k=0}^{n-1} \left[4\left(\sin^2\frac{(2j+1)\pi}{m}+\z^2\sin^2\frac{2\pi k}{n}\right)\right]^{1/2}\,.
\end{aligned}
\end{equation}
For odd $n$, using the identity $\sin^2(x+\pi)=\sin^2x$, we can 
rewrite the latter formula as
\begin{equation}\label{odd1}
\begin{aligned}
\Pf A_3&=\prod_{j=0}^{\frac{m}{2}-1}\prod_{k=0}^{n-1} \left[4\left(\sin^2\frac{(2j+1)\pi}{m}+\z^2\sin^2\frac{\pi k}{n}\right)\right]^{1/2}\,.
\end{aligned}
\end{equation}
Indeed, if we take $0\le k\le \frac{n-1}{2}\,$ in \eqref{tbc2.2},
then we obtain factors with  $\sin^2\frac{2\pi k}{n}$, while 
if we take $k=\frac{n+1}{2}+k'\,,$ $0\le k'\le \frac{n-3}{2}\,,$
then we obtain factors with 
\[
\sin^2\frac{2\pi}{n}\left(\frac{n+1}{2}+k'\right) =\sin^2\frac{\pi(2k'+1) }{n}\,.
\]
Combining these two cases, we obtain \eqref{odd1}. 

From \eqref{odd1}, using identity \eqref{eq3.2}, we obtain that
\begin{equation}\label{eq4.2}
\begin{aligned}
\Pf A_3=\prod_{k=0}^{n-1} 
\left[\left(u_{k}+\sqrt{1+u_{k}^2}\right)^{\frac{m}{2}}+\left(-u_{k}+\sqrt{1+u_{k}^2}\right)^{\frac{m}{2}}\right],
\end{aligned}
\end{equation}
where
\begin{equation}\label{eq4.21}
u_{k}=\z\sin(\pi x_{k})\ge0, \quad x_{k}=\frac{k}{n}\,,
\end{equation}
which implies that
\begin{equation}\label{eq4.3}
\ln (\Pf A_3)=G^{(3)}_{m,n}+H^{(3)}_{m,n}\,,
\end{equation}
where
\begin{equation}\label{eq4.4}
\begin{aligned}
&G^{(3)}_{m,n}=\frac{m}{2}\,\sum_{k=0}^{n-1} \ln\left(u_{k}+\sqrt{1+u_{k}^2}\right)
=\frac{m}{2}\,\sum_{k=0}^{n-1}  g\left(\frac{k}{n}\right),\\
&H^{(3)}_{m,n}=\sum_{k=0}^{n-1} \ln\left[1+\frac{1}{\left(u_{k}+\sqrt{1+u_{k}^2}\right)^m}\right].\\
\end{aligned}
\end{equation}
Formulae \eqref{eq4.2}-\eqref{eq4.4} are similar to \eqref{eq3.3}-\eqref{eq3.9} for even $n$,
but $x_k=\frac{2k}{n}$ for even $n$, while $x_k=\frac{k}{n}$ for odd $n$. This leads to the difference in elliptic nome \eqref{ntT6} for even and odd $n$.

Using the Euler--Maclaurin formula, we obtain  that
\begin{equation}\label{odd2}
G^{(3)}_{m,n}\sim SF-\frac{\ga}{12}
-\frac{m}{2}\sum_{p=1}^\infty\frac{B_{2p+2}\left(0\right)}{(p+1)n^{2p+1}}g_{2p+1}\,, \quad \ga=\pi\nu\z,
\end{equation}
and similar to the even case, we obtain the asymptotic expansion of 
$H^{(3)}_{m,n}$ as
\begin{equation}\label{odd3}
H^{(3)}_{m,n}\sim A^{(3)}+B^{(3)},
\end{equation}
with
\begin{equation}\label{odd4}
\begin{aligned}
&A^{(3)}=2\sum_{k=1}^\infty \ln(1+e^{-\ga k}),\\
&B^{(3)}=\frac{m}{2}\sum_{p=1}^\infty \frac{B_{2p+2}(0)g_{2p+1}}{(p+1)n^{2p+1}}+\sum_{p=1}^\infty\frac{R_p^{(3)}}{S^p}\,,
\end{aligned}
\end{equation}
where
\begin{equation}\label{odd5}
R_p^{(3)}=-\frac{\nu^{p+1}}{p+1}\,\De_p 
\left[K_{2p+2}^{0,\frac{1}{2}}
\left(\frac{i \nu\la}{2\pi}\right)\right]\bigg|_{\la=\pi\z}\,.
\end{equation}

Combining \eqref{eq4.3}, \eqref{odd2}, and \eqref{odd3}, we obtain 
that
\begin{equation}\label{eq4.17}
\begin{aligned}
\ln(\Pf A_3)&\sim SF-\frac{\ga}{12}+2\sum_{k=1}^\infty \ln(1+e^{-\ga k})+\sum_{p=1}^\infty\frac{R_p^{(3)}}{S^p}+\ln 2.
\end{aligned}
\end{equation}
Let now
\begin{equation}\label{eq4.18}
q=e^{\pi i\tau}=e^{-\frac{\ga}{2}}\,,
\end{equation}
then 
\begin{equation}\label{eq4.19}
\begin{aligned}
2e^{-\frac{\ga}{12}}\prod_{k=1}^\infty \left(1+e^{-\ga k}\right)^2&=2q^{1/6}\prod_{k=1}^{\infty}(1+q^{2k})^2
=2\left[\prod_{k=1}^{\infty}\frac{q^{1/6}(1-q^{4k})}{q^{1/12}(1-q^{2k})}\right]^2\\
&=2\left[\frac{\eta(2\tau)}{\eta(\tau)}\right]^2=\frac{\theta_2}{\eta}\,,
\end{aligned}
\end{equation}
hence
\begin{equation}\label{eq4.20}
\begin{aligned}
\ln Z=\ln(\Pf A_3)&\sim SF+\ln \frac{\theta_2}{\eta}
+\sum_{p=1}^\infty\frac{R_p^{(3)}}{S^p}\,.
\end{aligned}
\end{equation}
This finishes the proof of Theorem \ref{main_thmT_TBC}.

\begin{appendix}



\section{Exponent of a Taylor Series} \label{appA}

\begin{prop} We have that
\begin{equation}\label{exp1}
\begin{aligned}
\exp\left(\sum_{p=1}^\infty a_px^p\right)=1+\sum_{p=1}^\infty b_px^p,
\end{aligned}
\end{equation}
where
\begin{equation}\label{exp2}
\begin{aligned}
b_p= \sum_{\mathcal S_p}
\frac{(a_{p_1})^{q_1}\ldots (a_{p_r})^{q_r}}{q_1!\ldots q_r!}\,,
\end{aligned}
\end{equation}
and $\mathcal S_p$ is the set of collections of positive integers $(p_1,\ldots,p_r;q_1,\ldots,q_r)$, $1\le r\le p$, such that
\begin{equation}\label{exp3}
\begin{aligned}
\mathcal S_p=\left\{ 
(p_1,\ldots,p_r;q_1,\ldots,q_r)\;\big|\;
0<p_1<\ldots<p_r;\;  p_1q_1+\ldots+p_r q_r=p\right\}.
\end{aligned}
\end{equation}
The series in \eqref{exp1} are understood as formal ones.
\end{prop}

\begin{proof} Expanding the exponent into the Taylor series, we obtain that
\begin{equation}\label{exp4}
\begin{aligned}
\exp\left(\sum_{p=1}^\infty a_px^p\right)=1+\sum_{k=1}^\infty \frac{1}{k!}\left(\sum_{p=1}^\infty a_p x^p\right)^k.
\end{aligned}
\end{equation}
By the multinomial formula,
\begin{equation}\label{exp5}
\begin{aligned}
\frac{1}{k!}\left(\sum_{p=1}^\infty a_p x^p\right)^k
=\sum_{0<p_1<\ldots<p_r,\; q_1>0,\ldots, q_r>0:\; q_1+\ldots+q_r=k} 
\frac{(a_{p_1}x^{p_1})^{q_1}\ldots (a_{p_r}x^{p_r})^{q_r}}{q_1!\ldots q_r!},
\end{aligned}
\end{equation}
hence
\begin{equation}\label{exp6}
\begin{aligned}
\exp\left(\sum_{p=1}^\infty a_p x^p\right)&=1+\sum_{k=1}^\infty 
\sum_{0<p_1<\ldots<p_r,\; q_1>0,\ldots, q_r>0:\; q_1+\ldots+q_r=k} 
\frac{(a_{p_1}x^{p_1})^{q_1}\ldots (a_{p_r}x^{p_r})^{q_r}}{q_1!\ldots q_r!}\\
&=1+\sum_{r=1}^\infty 
\sum_{0<p_1<\ldots<p_r,\; q_1>0,\ldots, q_r>0} 
\frac{(a_{p_1}x^{p_1})^{q_1}\ldots (a_{p_r}x^{p_r})^{q_r}}{q_1!\ldots q_r!}.
\end{aligned}
\end{equation}
Combining terms with $ p_1q_1+\ldots+p_r q_r=p$, we obtain formulae \eqref{exp1}, \eqref{exp2}.
\end{proof}

\section{Bernoulli's Polynomials} \label{appC}

Bernoulli's polynomials are defined recursively by the equations,
\begin{equation}\label{bp1}
\begin{aligned}
B_k'(x)=kB_{k-1}(x),\quad \int\limits_0^1 B_k(x)\,dx=0,\quad k=1,2,\ldots;\quad B_0(x)=1.
\end{aligned}
\end{equation}
In particular,
\begin{equation}\label{bp1a}
\begin{aligned}
B_1(x)=x-\frac{1}{2}\,,\quad B_2(x)=x^2-x+\frac{1}{6}\,,\quad 
B_3(x)=x^3-\frac{3x^2}{2}+\frac{x}{2}\,.
\end{aligned}
\end{equation}
The Bernoulli periodic functions $\widehat B_k(x)$ are defined by
the periodicity condition $\widehat B_k(x+1)=\widehat B_k(x)$ and
by the condition $\widehat B_k(x)=B_k(x)$ for $0\le x\le 1$. 
Their Fourier series is equal to
\begin{equation}\label{bp2}
\begin{aligned}
\widehat B_k(x)=-\frac{k!}{(-2\pi i)^k} \sum_{\ell\not=0} \frac{e^{-2\pi i \ell x}}{\ell^k}\,.
\end{aligned}
\end{equation}
For $k\ge 2$ the Fourier series is absolutely convergent, and for $k=1$ it converges in $L^2[0,1]$.

The generating function of the Bernoulli polynomials is
\begin{equation}\label{bp7}
\begin{aligned}
G(\la; x):=\frac{\la\, e^{\la x}}{e^{\la}-1}=\sum_{k=0}^\infty \frac{\la^k B_k(x)}{k!}\,.
\end{aligned}
\end{equation}
Substituting \eqref{bp2} into \eqref{bp7}, we obtain that for $0\le x\le 1$ and $|\la|<2\pi$,
\begin{equation}\label{bp12}
\begin{aligned}
\frac{\la\, e^{\la x}}{e^{\la}-1}&-1-\la\left(x-\frac{1}{2}\right)=\sum_{k=2}^\infty \frac{\la^k B_k(x)}{k!}
=-\sum_{k=2}^\infty \sum_{\ell\not=0}\frac{\la^k }{(-2\pi i)^k}\,  \frac{e^{-2\pi i \ell x}}{\ell^k}\\
&=-\sum_{\ell\not=0}e^{-2\pi i \ell x}\sum_{k=2}^\infty \left(\frac{\la }{-2\pi i \ell}\right)^k  
=-\sum_{\ell\not=0}e^{-2\pi i \ell x}\left(\frac{\la}{-2\pi i \ell}\right)^2\frac{1}{1+\frac{\la}{2\pi i \ell}}\\
&=\frac{\la^2}{4\pi^2}\sum_{\ell\not=0}\frac{e^{-2\pi i \ell x}}{\ell\left(\ell+\frac{\la}{2\pi i }\right)}\,.
\end{aligned}
\end{equation}
Taking $\la=2\pi i z$, where $|z|<1$, $x=\al$, and $\ell=k$, we obtain that
\begin{equation}\label{bp13}
\begin{aligned}
\frac{2\pi i z e^{2\pi i z \al}}{e^{2\pi i z}-1}
=1+2\pi iz\left(\al-\frac{1}{2}\right)-z^2\sum_{k\not=0}\frac{e^{-2\pi i k\al}}{k\left(k+z\right)}\,,
\end{aligned}
\end{equation}
or
\begin{equation}\label{bp14}
\begin{aligned}
\frac{ e^{2\pi i z \al}}{e^{2\pi i z}-1}
=\frac{1}{2\pi i z}+\left(\al-\frac{1}{2}\right)-\frac{z}{2\pi i}\sum_{k\not=0}\frac{e^{-2\pi i k\al}}{k\left(k+z\right)}\,,
\end{aligned}
\end{equation}
We can rewrite the latter equation as 
\begin{equation}\label{bp15}
\begin{aligned}
\frac{ e^{2\pi i z \al}}{e^{2\pi i z}-1}
=\frac{1}{2\pi i z}+\left(\al-\frac{1}{2}\right)-\frac{1}{2\pi i}\sum_{k\not=0}e^{-2\pi i k\al}\left(\frac{1}{k}-\frac{1}{k+z}\right)\,,
\end{aligned}
\end{equation}

\section{The Euler--Maclaurin Formula} \label{appD}

Let $f(x)$ be an analytic function on the interval $[a,b]$. We partition the interval $[a,b]$
into $N$ equal intervals of the length
\begin{equation}\label{em1}
h=\frac{b-a}{N}\,.
\end{equation}
Let
\begin{equation}\label{em2}
x_k=a+kh+\al h,\quad k=0,1,\ldots, N-1,
\end{equation}
where $0\le \al\le 1$.
Then the Euler--Maclaurin formula with a remainder is
\begin{equation}\label{em3}
\begin{aligned}
\sum_{k=0}^{N-1} f(x_k)
&=\frac{1}{h}\int\limits_a^b f(x)\,dx
+\sum_{p=1}^\ell\frac{B_{p}(\al)h^{p-1}}{p!}[f^{(p-1)}(b)-f^{(p-1)}(a)]+R_{\ell}(\al), 
\end{aligned}
\end{equation}
where $B_p(\al)$ is the  Bernoulli polynomial and the remainder $R_{\ell}(\al)$ can be written as
\begin{equation}\label{em4}
\begin{aligned}
R_{\ell}(\al)= \frac{h^{\ell}}{\ell !}\int\limits_0^1 \widehat B_{\ell}(\al-\tau)\left[ \sum_{k=0}^{N-1} f^{(\ell)}(a+kh+\tau h)\right]\,d\tau,
\end{aligned}
\end{equation}
where $\widehat B_{\ell}(x)$ is the periodic Bernoulli  function. 

Thus, the Euler-Maclaurin formula gives an asymptotic series,
\begin{equation}\label{em5}
\begin{aligned}
\sum_{k=0}^{N-1} f(x_k)
&\sim \frac{1}{h}\int\limits_a^b f(x)\,dx
+\sum_{p=1}^\infty\frac{B_{p}(\al)h^{p-1}}{p!}[f^{(p-1)}(b)-f^{(p-1)}(a)].
\end{aligned}
\end{equation}
In general, since both $B_p(x)$ and $f^{(p)}(x)$ grow
like $p!$, the series on the right in \eqref{em5} diverges.

\section{Kronecker's Double Series of Pure Imaginary Argument}\label{appF}

A classical reference to the Kronecker double series is the book of Weil \cite{Weil}. In this Appendix we
review and specify some results of  Ivashkevich,  Izmailian, and Hu \cite{IIH}. 
Let us consider the Kronecker double series with parameters $(\alpha,\beta)=(\frac{1}{2},0)$ as defined in \eqref{eq2.7} with argument $\tau=2ir\z$ and $(\al,\be)=(\frac{1}{2},\frac{1}{2}),(0,\frac{1}{2})$ as defined in \eqref{eq2.7} with argument $\tau=ir\z$. Observe that in all cases, if $p$ is odd, the terms $(j,k)$ and $(-j,-k)$ cancel each other. Hence $K^{\frac{1}{2},0}_{2p-1}(\tau)=K^{\frac{1}{2},\frac{1}{2}}_{2p-1}(\tau)=K^{0,\frac{1}{2}}_{2p-1}(\tau)=0$ for $p=1,2,\ldots\ $. Therefore, we will take $p$ to be even. Let us first consider the case $(\alpha,\beta)=(\frac{1}{2},0)$.\\
\subsection{Case $(\alpha,\beta)=(\frac{1}{2},0)$} From \eqref{eq2.7}, we have that
\begin{equation}\label{kr1}
\begin{aligned}
K_{2p}^{\frac{1}{2},0}(\tau)=\frac{(-1)^{p+1} (2p)!}{(2\pi )^{2p}}\sum_{(j,k)\not=(0,0)}\frac{(-1)^{k}}{(k+\tau j)^{2p}}\,.
\end{aligned}
\end{equation}
Separating terms with $j=0$, we obtain that
\begin{equation}\label{kr2}
\begin{aligned}
K_{2p}^{\frac{1}{2},0}(\tau)=\frac{(-1)^{p+1} (2p)!}{(2\pi )^{2p}}\sum_{k\not=0}\frac{(-1)^k}{k^{2p}}
+\frac{(-1)^{p+1} (2p)!}{(2\pi )^{2p}}\sum_{j\not=0}\sum_{k=-\infty}^\infty\frac{(-1)^{k}}{(k+\tau j)^{2p}}\,.
\end{aligned}
\end{equation}
The first term is just the Fourier series for the Bernoulli polynomial $B_{2p}\left(x\right)$ evaluated at $x=\frac{1}{2}$. Let us transform the second term.
Since the terms $(j,k)$ and $(-j,-k)$ give the same contribution, we can write that
\begin{equation}\label{kr3}
\begin{aligned}
K_{2p}^{\frac{1}{2},0}(\tau)=B_{2p}\left(\frac{1}{2}\right)
+\frac{2(-1)^{p+1} (2p)!}{(2\pi )^{2p}}\sum_{j=1}^\infty\sum_{k=-\infty}^\infty\frac{(-1)^{k}}{(k+\tau j)^{2p}}\,.
\end{aligned}
\end{equation}
When $z=iy$, $y>0$, and $\al=1/2$,
identity \eqref{bp15} reads
\begin{equation}\label{kr4}
\begin{aligned}
\frac{ e^{-\pi y }}{1-e^{-2\pi y}}
=\frac{1}{2\pi y}+\frac{1}{2\pi i}\sum_{k\not=0}(-1)^k\left(\frac{1}{k}-\frac{1}{k+iy}\right)\,.
\end{aligned}
\end{equation}
Expanding the left hand side into the geometric series, we obtain that
\begin{equation}\label{kr5}
\begin{aligned}
\sum_{k=0}^\infty e^{-2\pi y(k+\frac{1}{2})}
=\frac{1}{2\pi y}+\frac{1}{2\pi i}\sum_{k\not=0}(-1)^k\left(\frac{1}{k}-\frac{1}{k+iy}\right)\,.
\end{aligned}
\end{equation}
Differentiating this identity $(2p-1)$ times with respect to $y$, we obtain that
\begin{equation}\label{kr6}
\begin{aligned}
\sum_{k=0}^\infty \left[-2\pi\left(k+\frac{1}{2}\right)\right]^{2p-1}e^{-2\pi y(k+\frac{1}{2})}
=-\frac{1}{2\pi i}\sum_{k=-\infty}^\infty \frac{(-i)^{2p-1}(2p-1)!(-1)^k}{(k+iy)^{2p}}\,,
\end{aligned}
\end{equation}
or equivalently,
\begin{equation}\label{kr7}
\begin{aligned}
\frac{(-1)^p(2p)!}{(2\pi)^{2p}}\sum_{k=-\infty}^\infty \frac{(-1)^k}{(k+iy)^{2p}}
=2p \sum_{k=0}^\infty \left(k+\frac{1}{2}\right)^{2p-1}e^{-2\pi y(k+\frac{1}{2})}.
\end{aligned}
\end{equation}
Using this formula in \eqref{kr3} with $y=r\z j$, we obtain that
\begin{equation}\label{kr8}
\begin{aligned}
K_{2p}^{\frac{1}{2},0}(\tau)=B_{2p}\left(\frac{1}{2}\right)
-4p \sum_{j=1}^\infty \sum_{k=0}^\infty \left(k+\frac{1}{2}\right)^{2p-1}e^{-2r\la j(k+\frac{1}{2})}\,.
\end{aligned}
\end{equation}
Thus, we have the following proposition:

\begin{prop} We have that
\begin{equation}\label{kr9}
\begin{aligned}
 \sum_{j=1}^\infty \sum_{k=0}^\infty \left(k+\frac{1}{2}\right)^{2p-1}e^{-2r\la j(k+\frac{1}{2})}
=\frac{B_{2p}\left(\frac{1}{2}\right)-K_{2p}^{\frac{1}{2},0}(\tau)}{4p}\,.
\end{aligned}
\end{equation}
\end{prop}
\noindent Applying it for $K_{2p+2}^{\frac{1}{2},0}(\tau)$, we obtain that
\begin{equation}\label{kr10}
\begin{aligned}
 \sum_{j=1}^\infty \sum_{k=0}^\infty \left(k+\frac{1}{2}\right)^{2p+1}e^{-2r\la j(k+\frac{1}{2})}
=\frac{B_{2p+2}\left(\frac{1}{2}\right)-K_{2p+2}^{\frac{1}{2},0}(\tau)}{4(p+1)}\,.
\end{aligned}
\end{equation}

\subsection{Case $(\alpha,\beta)=(\frac{1}{2},\frac{1}{2})$}
From \eqref{eq2.7}, we have that
\begin{equation}\label{kr11}
\begin{aligned}
K_{2p}^{\frac{1}{2},\frac{1}{2}}(\tau)=\frac{(-1)^{p+1} (2p)!}{(2\pi )^{2p}}\sum_{(j,k)\not=(0,0)}\frac{(-1)^{k+j}}{(k+\tau j)^{2p}}\,.
\end{aligned}
\end{equation}
Separating terms with $j=0$, we obtain that
\begin{equation}\label{kr12}
\begin{aligned}
K_{2p}^{\frac{1}{2},\frac{1}{2}}(\tau)=\frac{(-1)^{p+1} (2p)!}{(2\pi )^{2p}}\sum_{k\not=0}\frac{(-1)^k}{k^{2p}}
+\frac{(-1)^{p+1} (2p)!}{(2\pi )^{2p}}\sum_{j\not=0}\sum_{k=-\infty}^\infty\frac{(-1)^{k+j}}{(k+\tau j)^{2p}}\,.
\end{aligned}
\end{equation}
The first term is just the Fourier series for the Bernoulli polynomial $B_{2p}\left(x\right)$ evaluated at $x=\frac{1}{2}$. Let us transform the second term. Since the terms $(j,k)$ and $(-j,-k)$ give the same contribution, we can write that
\begin{equation}\label{kr13}
\begin{aligned}
K_{2p}^{\frac{1}{2},\frac{1}{2}}(\tau)=B_{2p}\Big(\frac{1}{2}\Big)
+\frac{2(-1)^{p+1} (2p)!}{(2\pi )^{2p}}\sum_{j=1}^\infty\sum_{k=-\infty}^\infty\frac{(-1)^{k+j}}{(k+\tau j)^{2p}}\,.
\end{aligned}
\end{equation}
Now from identity \eqref{kr7}, with $y=r\z j$, we obtain that
\begin{equation}\label{kr14}
\begin{aligned}
K_{2p}^{\frac{1}{2},\frac{1}{2}}(\tau)=B_{2p}\Big(\frac{1}{2}\Big)
-4p \sum_{j=1}^\infty \sum_{k=0}^\infty (-1)^j\left(k+\frac{1}{2}\right)^{2p-1}e^{-2r\la j(k+\frac{1}{2})}\,.
\end{aligned}
\end{equation}
Thus, we have the following proposition:

\begin{prop} We have that
\begin{equation}\label{kr15}
\begin{aligned}
 \sum_{j=1}^\infty \sum_{k=0}^\infty (-1)^j\left(k+\frac{1}{2}\right)^{2p-1}e^{-2r\la j(k+\frac{1}{2})}
=\frac{B_{2p}\left(\frac{1}{2}\right)-K_{2p}^{\frac{1}{2},\frac{1}{2}}(\tau)}{4p}\,.
\end{aligned}
\end{equation}
\end{prop}
\noindent Applying it for $K_{2p+2}(\tau)$, we obtain that
\begin{equation}\label{kr16}
\begin{aligned}
 \sum_{j=1}^\infty \sum_{k=0}^\infty (-1)^j\left(k+\frac{1}{2}\right)^{2p+1}e^{-2r\la j(k+\frac{1}{2})}
=\frac{B_{2p+2}\left(\frac{1}{2}\right)-K_{2p+2}^{\frac{1}{2},\frac{1}{2}}(\tau)}{4(p+1)}\,.
\end{aligned}
\end{equation}

\subsection{Case $(\alpha,\beta)=(0,\frac{1}{2})$}
From \eqref{eq2.7}, we have that
\begin{equation}\label{kr17}
\begin{aligned}
K_{2p}^{0,\frac{1}{2}}(\tau)=\frac{(-1)^{p+1} (2p)!}{(2\pi )^{2p}}\sum_{(j,k)\not=(0,0)}\frac{(-1)^{j}}{(k+\tau j)^{2p}}\,.
\end{aligned}
\end{equation}
Separating terms with $j=0$, we obtain that
\begin{equation}\label{kr18}
\begin{aligned}
K_{2p}^{0,\frac{1}{2}}(\tau)=\frac{(-1)^{p+1} (2p)!}{(2\pi )^{2p}}\sum_{k\not=0}\frac{1}{k^{2p}}
+\frac{(-1)^{p+1} (2p)!}{(2\pi )^{2p}}\sum_{j\not=0}\sum_{k=-\infty}^\infty\frac{(-1)^{j}}{(k+\tau j)^{2p}}\,.
\end{aligned}
\end{equation}
The first term is just the Fourier series for the Bernoulli polynomial $B_{2p}\left(x\right)$ evaluated at $x=0$. Let us transform the second term.
Since the terms $(j,k)$ and $(-j,-k)$ give the same contribution, we can write that
\begin{equation}\label{kr19}
\begin{aligned}
K_{2p}^{0,\frac{1}{2}}(\tau)=B_{2p}\left(0\right)
+\frac{2(-1)^{p+1} (2p)!}{(2\pi )^{2p}}\sum_{j=1}^\infty\sum_{k=-\infty}^\infty\frac{(-1)^{j}}{(k+\tau j)^{2p}}\,.
\end{aligned}
\end{equation}
When $z=iy$, $y>0$, and $\al=0$,
identity \eqref{bp15} reads
\begin{equation}\label{kr20}
\begin{aligned}
\frac{1}{1-e^{-2\pi y}}
=\frac{1}{2\pi y}+\frac{1}{2}+\frac{1}{2\pi i}\sum_{k\not=0}\left(\frac{1}{k}-\frac{1}{k+iy}\right)\,.
\end{aligned}
\end{equation}
Expanding the left hand side into the geometric series, we obtain that
\begin{equation}\label{kr21}
\begin{aligned}
\sum_{k=0}^\infty e^{-2\pi yk}
=\frac{1}{2\pi y}+\frac{1}{2}+\frac{1}{2\pi i}\sum_{k\not=0}\left(\frac{1}{k}-\frac{1}{k+iy}\right)\,.
\end{aligned}
\end{equation}
Differentiating this identity $(2p-1)$ times with respect to $y$, we obtain that
\begin{equation}\label{kr22}
\begin{aligned}
\sum_{k=0}^\infty \left[-2\pi k\right]^{2p-1}e^{-2\pi yk}
=-\frac{1}{2\pi i}\sum_{k=-\infty}^\infty \frac{(-i)^{2p-1}(2p-1)!}{(k+iy)^{2p}}\,,
\end{aligned}
\end{equation}
or equivalently,
\begin{equation}\label{kr23}
\begin{aligned}
\frac{(-1)^p(2p)!}{(2\pi)^{2p}}\sum_{k=-\infty}^\infty \frac{1}{(k+iy)^{2p}}
=2p \sum_{k=0}^\infty k^{2p-1}e^{-2\pi yk}.
\end{aligned}
\end{equation}
Using this formula in \eqref{kr19} with $y=r\z j$, we obtain that
\begin{equation}\label{kr24}
\begin{aligned}
K_{2p}^{0,\frac{1}{2}}(\tau)=B_{2p}\left(0\right)
-4p \sum_{j=1}^\infty \sum_{k=0}^\infty (-1)^jk^{2p-1}e^{-2\la r jk}\,.
\end{aligned}
\end{equation}
Thus, we have the following proposition:

\begin{prop} We have that
\begin{equation}\label{kr25}
\begin{aligned}
 \sum_{j=1}^\infty \sum_{k=0}^\infty (-1)^jk^{2p-1}e^{-2\la r jk}
=\frac{B_{2p}\left(0\right)-K_{2p}^{0,\frac{1}{2}}(\tau)}{4p}\,.
\end{aligned}
\end{equation}
\end{prop}

Applying it for $K_{2p+2}^{0,\frac{1}{2}}(\tau)$, we obtain that
\begin{equation}\label{kr26}
\begin{aligned}
 \sum_{j=1}^\infty \sum_{k=0}^\infty (-1)^jk^{2p+1}e^{-2\la r jk}
=\frac{B_{2p+2}\left(0\right)-K_{2p+2}^{0,\frac{1}{2}}(\tau)}{4(p+1)}\,.
\end{aligned}
\end{equation}

\end{appendix}

\end{document}